\documentclass[12pt]{article}
\pdfoutput=1

\usepackage{xcolor}
\definecolor{blue}{RGB}{0,70,140}
\definecolor{green}{RGB}{100,140,0}
\definecolor{red}{RGB}{190,10,50}

\colorlet{structurecolor}{red}
\colorlet{linkcolor}{blue}
\colorlet{citecolor}{green}

\usepackage[utf8]{inputenc}
\usepackage[T1]{fontenc}
\usepackage{mathpazo}

\usepackage[english]{babel}

\usepackage[a4paper,margin=2.5cm]{geometry}
\usepackage{parskip}

\newcommand{\affiliation}{\footnote}
\newcommand{\affiliationmark}[1][\value{footnote}-1]{\footnotemark[\numexpr#1+1\relax]}
\usepackage[shortlabels]{enumitem}
\setlist{itemsep=0ex,topsep=0ex,parsep=0.4ex}

\usepackage{amsmath,amsfonts,amssymb,amsthm,mathtools,thmtools,thm-restate,bbm,centernot}

\usepackage[hyphens]{url} 
\usepackage[hidelinks,colorlinks,pagebackref]{hyperref} 
\hypersetup{citecolor=citecolor,linkcolor=linkcolor,urlcolor=linkcolor}
\usepackage[capitalise,nameinlink,noabbrev,compress]{cleveref} 

\renewcommand*{\backref}[1]{}
\renewcommand*{\backrefalt}[4]{
	\ifcase #1 Not cited.%
	\or $\uparrow$#2%
	\else $\uparrow$#2%
	\fi%
}

\declaretheorem[name=Definition,Refname={Definition,Definitions},numberwithin=section,style=definition]{definition}
\declaretheorem[name=Theorem,Refname={Theorem,Theorems},numberlike=definition,style=plain]{theorem}

\declaretheorem[name=Lemma,Refname={Lemma,Lemmas},numberlike=theorem,style=plain]{lemma}

\newcommand{\defn}[1]{\textcolor{structurecolor}{\emph{#1}}}

\mathtoolsset{centercolon}

\renewcommand{\l}{\mathopen{}\mathclose\bgroup\left}
\renewcommand{\r}{\aftergroup\egroup\right}

\newcommand{\st}{\ifnum\currentgrouptype=16 \mathrel{}\middle|\mathrel{}\else\mathrel{|}\fi}

\DeclarePairedDelimiter{\set}{\{}{\}}
\DeclarePairedDelimiter{\abs}{\lvert}{\rvert}

\DeclarePairedDelimiter{\floor}{\lfloor}{\rfloor}

\newcommand{\subs}{\subseteq}

\newcommand{\sm}{\setminus}
\newcommand{\eps}{\varepsilon}

\newcommand{\ub}[1]{\underbrace{#1}}

\makeatletter
\newcommand{\generate}[4]{%
  \def\@tempa{#1} 
  \count@=`#3
  \loop
  \begingroup\lccode`?=\count@
  \lowercase{\endgroup\@namedef{\@tempa ?}{#2{?}}}%
  \ifnum\count@<`#4
  \advance\count@\@ne
  \repeat
}
\makeatother

\generate{b}{\mathbb}{A}{Z}
\generate{c}{\mathcal}{A}{Z}

\newcommand{\DeclareMath}[2]{\newcommand{#1}{\mathnormal{#2}}}

\newcommand{\DeclareMathObject}[2]{\newcommand{#1}{\mathrm{#2}}}
\newcommand{\RedeclareMathObject}[2]{\renewcommand{#1}{\mathrm{#2}}}
\newcommand{\DeclareBinaryMathOperator}[2]{\newcommand{#1}{\mathbin{#2}}}

\DeclareMath{\N}{\mathbb{N}}
\DeclareMath{\Z}{\mathbb{Z}}
\DeclareMath{\Q}{\mathbb{Q}}
\DeclareMath{\R}{\mathbb{R}}
\DeclareMath{\C}{\mathbb{C}}
\DeclareMath{\F}{\mathbb{F}}

\DeclareMath{\ind}{\mathbbm{1}}
\DeclareBinaryMathOperator{\divides}{|}

\DeclareMathOperator{\pr}{\bP}
\DeclareMathOperator{\ex}{\bE}

\RedeclareMathObject{\P}{P}
\DeclareMathObject{\PLS}{PLS}

\newcommand{\restr}[2]{#1|_{#2}}
\newcommand{\impr}[2]{\Delta^{#2}(#1)}
\newcommand{\aperr}{\eps}

\RequirePackage{tikz}
\usetikzlibrary{calc,arrows.meta,graphs,quotes,backgrounds}

\newenvironment{graph}[1][]{\begin{tikzpicture}[>=Latex,#1]}{\end{tikzpicture}}

\newcommand{\vertex}[3][]{\node[circle,fill,inner sep=0cm,minimum size=0.15cm,#1] (#2) at (#3) {};}
\newcommand{\edges}[2][]{\edef\drawedges{\noexpand\graph[#1]{#2};}\drawedges}

\RequirePackage{listofitems}

\pgfkeys{
  expex/.is family,expex
}
\newif\ifexpexdirected
\pgfkeys{
  expex,
  name/.estore in=\expexname,
  shift/.estore in=\expexshift,
  scale/.estore in=\expexscale,
  vertex labels/.code=\readlist\expexvertexlabel{#1},
  vertex labels macro/.code={\def\expexvertexlabel[##1]{#1}},
  highlight color/.estore in=\expexhighlightcolor,
  highlight vertex/.estore in=\expexhighlightvertex,
  edge label/.code=\def\expexedgelabel{#1},
  edge labels/.code=\readlist\expexedgelabel{#1},
  connect edge label/.style={edge labels={1,-1,#1,-#1}},
  directed/.is if=expexdirected,
  state/.estore in=\expexstate,
  default/.style={
    shift={(0,0)},
    scale=1,
    name=,
    vertex labels macro=,
    highlight color=highlightcolor,
    highlight vertex=-1,
    edge labels={,,,},
    directed=false,
    state=0
  }
}
\newcommand{\expexstsel}[3]{\ifnum\expexstate<#1 #2\else#3\fi}
\newcommand{\expexstsele}[4]{\ifnum\expexstate<#1 #2\else\ifnum\expexstate=#1 #3\else#4\fi\fi}
\newcommand{\expexvertcol}[1]{\ifnum\expexhighlightvertex=#1 \expexhighlightcolor\else black\fi}
\newcommand{\expexvert}[4]{
  \vertex["{$\expexvertexlabel[#1]$}"{#4},\expexvertcol{#1}]{#1}{#2,\expexstsel{#1}{#3}{1-#3}}
  \ifnum\expexhighlightvertex=#1
    \draw[\expexhighlightcolor,->] (#1) -- (#2,\expexstsel{#1}{1-#3}{#3});
  \fi
}
\newcommand{\expexedge}{\ifexpexdirected->\else--\fi}
\newcommand{\expexbase}[1][]{
  \pgfkeys{expex,default,edge label=,#1}
  \begin{scope}[name prefix=\expexname,shift={(\expexshift)},scale=\expexscale]
    \expexvert{8}{0}{0}{below}
    \expexvert{1}{1}{1}{}
    \edges[edge quotes={auto,inner sep=0.05cm}]{(8) \expexedge["$\expexedgelabel$"'pos=\expexstsele{1}{0.4}{0.5}{0.65}] (1)}
  \end{scope}
}
\newcommand{\expexblackbox}[1][]{
  \pgfkeys{expex,default,#1}
  \begin{scope}[shift={(\expexshift)},scale=\expexscale,name prefix=\expexname]
    \vertex["{$\expexvertexlabel[8]$}"{below},\expexvertcol{1}]{8}{0,\expexstsel{8}{0}{1-0}}
    \expexvert{1}{3}{1}{}
    \edges[no placement,edge quotes={auto,inner sep=0.05cm}]{
      /[coordinate,at={(1,\expexstsel{7}{1}{0})},inner sep=0cm]
          --[dashed,\expexvertcol{1}] /[coordinate,at={(0.5,0.5)}]
          --["$\expexedgelabel[1]$"pos=0.2,\expexvertcol{1}] (8),
      /[coordinate,at={(2,\expexstsel{2}{0}{1})}]
          --[dashed,\expexvertcol{1}] /[coordinate,at={(2.5,0.5)}]
          \expexedge["$\expexedgelabel[2]$"inner sep=0cm,\expexvertcol{1}] (1)
    }
    \node[\expexvertcol{1}] at (1.5,0.5) {\tiny$\bullet\bullet\bullet$};
  \end{scope}
}
\newcommand{\expexgadget}[1][]{
  \pgfkeys{expex,default,#1}
  \begin{scope}[shift={(\expexshift)},scale=\expexscale,name prefix=\expexname]
    \expexvert{8}{0}{0}{below}
    \expexvert{7}{1}{1}{above right,xshift=-0.2cm}
    \expexvert{6}{2}{0}{below right,xshift=-0.2cm}
    \expexvert{5}{3}{1}{above right,xshift=-0.2cm}
    \expexvert{4}{4}{0}{below right,xshift=-0.2cm}
    \expexvert{3}{5}{1}{above right,xshift=-0.2cm}
    \expexvert{2}{6}{0}{below right,xshift=-0.2cm}
    \expexvert{1}{7}{1}{}
  \end{scope}
  \edges[edge quotes={auto,inner sep=0.05cm}]{
    (\expexname 8) \expexedge["$\expexedgelabel[1]$"'pos=\expexstsele{7}{0.4}{0.5}{0.65}] (\expexname 7)
        \expexedge["$\expexedgelabel[1]$"pos=\expexstsele{6}{0.4}{0.5}{0.65}] (\expexname 6)
        \expexedge["$\expexedgelabel[2]$"'pos=\expexstsele{5}{0.4}{0.5}{0.65}] (\expexname 5)
        \expexedge["$\expexedgelabel[2]$"pos=\expexstsele{4}{0.4}{0.5}{0.65}] (\expexname 4)
        \expexedge["$\expexedgelabel[3]$"'pos=\expexstsele{3}{0.4}{0.5}{0.65}] (\expexname 3)
        \expexedge["$\expexedgelabel[3]$"pos=\expexstsele{2}{0.4}{0.5}{0.65}] (\expexname 2)
        \expexedge["$\expexedgelabel[4]$"'pos=\expexstsele{1}{0.4}{0.5}{0.65}] (\expexname 1)
  }
}
\newcommand{\expexgadetconnect}[3][]{
  \pgfkeys{expex,default,#1}
  \edges[edge quotes={auto,inner sep=0.05cm}]{
    (#28) \expexedge[bend left=\expexstsel{7}{55}{70},looseness=\expexstsel{7}{1.3}{1.5},"$\expexedgelabel[1]$"near end] (#37),
    (#28) \expexedge[bend left=\expexstsel{5}{65}{75},looseness=\expexstsel{5}{1.2}{1.4},"$\expexedgelabel[2]$"near end] (#35),
    (#28) \expexedge[bend left=\expexstsel{3}{70}{80},looseness=\expexstsel{3}{1.2}{1.3},"$\expexedgelabel[1]$"near end] (#33),
    (#21) \expexedge[bend right=\expexstsel{6}{55}{70},looseness=\expexstsel{6}{1.3}{1.5},"$\expexedgelabel[3]$"'near end] (#36),
    (#21) \expexedge[bend right=\expexstsel{4}{65}{75},looseness=\expexstsel{4}{1.2}{1.4},"$\expexedgelabel[4]$"'near end] (#34),
    (#21) \expexedge[bend right=\expexstsel{2}{70}{80},looseness=\expexstsel{2}{1.2}{1.3},"$\expexedgelabel[3]$"'near end] (#32)
  }
}

\pgfkeys{
  tlmaxcut/.is family,tlmaxcut
}
\newif\iftlmaxcutimprovementedge
\newif\iftlmaxcuthighlight
\pgfkeys{
  tlmaxcut,
  shift/.estore in=\tlmaxcutshift,
  scale/.estore in=\tlmaxcutscale,
  vertex label/.store in=\tlmaxcutvertexlabel,
  improvement edge/.is if=tlmaxcutimprovementedge,
  edge label/.estore in=\tlmaxcutedgelabel,
  highlight/.is if=tlmaxcuthighlight,
  highlight color/.estore in=\tlmaxcuthighlightcolor,
  state/.estore in=\tlmaxcutstate,
  default/.style={
    shift={(0,0)},
    scale=1,
    vertex label=,
    improvement edge=false,
    edge label=,
    highlight=false,
    highlight color=highlightcolor,
    state=0
  }
}
\newcommand{\tlmaxcutsel}[3]{\ifnum\tlmaxcutstate=0 #1\else\ifnum\tlmaxcutstate=1 #2\else#3\fi\fi}

\newcommand{\tlmaxcuttree}[1][]{
  \pgfkeys{tlmaxcut,default,#1}
  \begin{scope}[shift={(\tlmaxcutshift)},scale=\tlmaxcutscale]
    \vertex["$\tlmaxcutvertexlabel$"below]{v1}{0,0}
    \foreach \a in {-160,-30,0,30,80,130,160} {
      \draw (v1) to (\a:0.8) edge (\a+10:1.5) edge (\a-10:1.5);
    }
  \end{scope}
}

\title{Superpolynomial smoothed complexity\\of 3-FLIP in Local Max-Cut}
\author{Lukas Michel\affiliation{Mathematical Institute, University of Oxford, United Kingdom (\textsf{\{\href{mailto:michel@maths.ox.ac.uk}{michel},\href{mailto:scott@maths.ox.ac.uk}{scott}\}@maths.ox.ac.uk}). Research of Alex Scott supported by EPSRC grant EP/V007327/1.}\qquad Alex Scott\affiliationmark[1]}
\date{26 September 2024}

\begin{document}
  \maketitle

  \begin{abstract}
    Local search algorithms for NP-hard problems such as Max-Cut frequently perform much better in practice than worst-case analysis suggests. Smoothed analysis has proved an effective approach to understanding this: a substantial literature shows that when a small amount of random noise is added to input data, local search algorithms typically run in polynomial or quasi-polynomial time. In this paper, we provide the first example where a local search algorithm for the Max-Cut problem fails to be efficient in the framework of smoothed analysis. Specifically, we construct a graph with $n$ vertices where the smoothed runtime of the 3-FLIP algorithm can be as large as $2^{\Omega(\sqrt{n})}$.
    
    Additionally, for the setting without random noise, we give a new construction of graphs where the runtime of the FLIP algorithm is $2^{\Omega(n)}$ for any pivot rule. These graphs are much smaller and have a simpler structure than previous constructions.
  \end{abstract}

  \section{Introduction}

  In the Max-Cut problem, the vertices of a weighted graph have to be partitioned into two sets such that the sum of the weights of all edges crossing between the sets is maximal. Since this problem is computationally hard to solve \cite{karp1972reducibility}, local search is often used to compute reasonably good solutions in an acceptable time. A standard form of local search for the Max-Cut problem is the FLIP algorithm which repeatedly moves individual vertices across the cut until it reaches a local optimum; similarly, the $k$-FLIP algorithm moves up to $k$ vertices in each step. The FLIP algorithm performs well in practice \cite{gosti1995approximationmaxcut, dolezal1999comparison}, despite the fact that there exist graphs where the algorithm takes exponential time to terminate \cite{schaffer1991localsearchhard, elsasser2011localmaxcut, monien2010degreefourlocalmaxcut}.  

  An important approach for closing the gap between the worst-case analysis and the empirical performance of algorithms is through \defn{smoothed analysis} \cite{spielman2004smoothedanalysis}. This method has been particularly effective for local search algorithms \cite{mathey2015smoothedanalysis}. In this framework, the runtime of the FLIP algorithm is analysed after a small amount of random noise is added to the edge weights of the graph. Recent work has shown that the smoothed runtime of the FLIP algorithm is polynomial or quasi-polynomial in various different settings \cite{etscheid2017smoothedlocalmaxcut, chen2020smoothedlocalmaxcut, elsasser2011localmaxcut, giannakopoulos2022smoothed, angel2017smoothedpolynomiallocalmaxcut, bibak2021smoothedflipmaxcut}. Similarly, the 2-FLIP algorithm has been shown to have a smoothed quasi-polynomial runtime in complete graphs \cite{chen2023smoothed}.

  In this paper, we consider the 3-FLIP algorithm. In light of the previous work, it might be expected that smoothed analysis should give the same picture, with a polynomial or near-polynomial smoothed runtime. We show that this is not the case: the smoothed runtime of the 3-FLIP algorithm in an $n$-vertex graph can be exponentially large in $\sqrt{n}$. 

  \begin{theorem}
    \label{thm:3flipsuppolyruntime}
    For all $n \in \N$ there exist graphs $G$ with $\cO(n)$ vertices such that the following holds. Let $a < b$ be real numbers. Suppose that the edge weights of $G$ are chosen uniformly at random in $[a, b]$ and consider a uniformly random initial cut of $G$. Then, with high probability there are executions of the 3-FLIP algorithm from this initial cut that take $2^{\Omega(\sqrt{n})}$ steps.
  \end{theorem}

  This is the first example of a local search algorithm for the Max-Cut problem whose smoothed runtime can be inefficient. We note that \cref{thm:3flipsuppolyruntime} is also quite robust: the weights can be smoothed across any interval, and the proofs could be adjusted to handle other smoothing distributions. While most implementations of the 3-FLIP algorithm will not follow the long local search sequences that we construct, this result shows that the details of the algorithm matter when performing smoothed analysis.

  Returning to the non-smoothed problem, we also give a new construction of graphs with maximum degree four where the FLIP algorithm takes exponential time to terminate. While graphs with these properties were previously constructed by Monien and Tscheuschner \cite{monien2010degreefourlocalmaxcut}, our construction has a simpler structure and uses fewer vertices. Additionally, there is exactly one execution of the FLIP algorithm in our graph that ends in a local optimum, while in previous constructions there were multiple.
  
  \begin{theorem}
    \label{thm:flipruntimeexp}
    For all $n \in \N$ there exist graphs with $\cO(n)$ vertices and maximum degree four which have an initial cut from which the unique execution of the FLIP algorithm that ends in a local optimum takes $\Omega(2^n)$ steps.
  \end{theorem}

  This paper is organized as follows. In \cref{ssec:localsearch,ssec:smoothedanalysis} we give some background on local search and smoothed analysis for the Max-Cut problem. We then provide the new construction of graphs with the properties from \cref{thm:flipruntimeexp} in \cref{sec:flipruntime}. \cref{thm:3flipsuppolyruntime} is proved in \cref{sec:3flipsmoothedruntime}. We conclude with some discussion in \cref{sec:openproblems}.

  \subsection{Local Search}
  \label{ssec:localsearch}

  The Max-Cut problem has found many practical applications, including circuit layout design \cite{barahona1988maxcutapplication, chen1983viaminimization, pinter1984layerassignment, chang1987layerassignment}, clustering \cite{poland2006clustering}, the analysis of newsgroups \cite{agrawal2003miningnewsgroups}, and more \cite{poljak1995maximumcuts, boros1991maxcut, chen2004hierarchicalsupportvectormachines, bramoulle2007anticoordination}. However, Max-Cut is NP-hard \cite{karp1972reducibility}, and so there likely exists no efficient algorithm solving this problem. Instead, in applications of Max-Cut, local search is often used to compute large cuts in a short amount of time.

  Local search algorithms are successful in a variety of optimization problems, such as linear programming \cite{spielman2004smoothedanalysis} or the travelling salesman problem \cite{johnson1997travelingsalesmanproblem}. For Max-Cut, a simple but successful form of local search is the FLIP algorithm, which starts with some initial cut and then repeatedly moves individual vertices across the cut until it reaches a local optimum. A pivot rule determines which vertex to move if there are multiple vertices that could be moved to yield a local improvement. \defn{Local Max-Cut} is the problem of finding a local optimum, that is a cut which cannot be improved by flipping the position of a single vertex. The Local Max-Cut problem is related to Hopfield neural networks \cite{hopfield1982neuralnetworks} and Nash equilibria in party affiliation games \cite{fabrikant2004nashequilibria}.

  While the FLIP algorithm performs well in practice \cite{gosti1995approximationmaxcut, dolezal1999comparison}, worst-case analysis cannot explain its good performance. Sch{\"a}ffer and Yannakakis \cite{schaffer1991localsearchhard} have shown that the Local Max-Cut problem is PLS-complete where PLS is the complexity class of all polynomial local search problems. This is already the case even for Local Max-Cut on graphs whose maximum degree is five \cite{elsasser2011localmaxcut}. Not only does this prove that the Local Max-Cut problem cannot be solved efficiently unless $\PLS \subs \P$, it also implies that there exist graphs where the FLIP algorithm takes exponential time to terminate. In fact, Monien and Tscheuschner \cite{monien2010degreefourlocalmaxcut} have shown that there are graphs of maximum degree four with this property. We provide a new construction of such graphs, proving \cref{thm:flipruntimeexp}. Note that the condition of having maximum degree four cannot be reduced since the FLIP algorithm runs in polynomial time on graphs whose maximum degree is three \cite{poljak1995localmaxcut}.

  \subsection{Smoothed Analysis}
  \label{ssec:smoothedanalysis}

  Besides the FLIP algorithm, there are also many other algorithms whose worst case analysis incorrectly predicts their real-world performance. To explain this difference for the simplex algorithm, Spielman and Teng \cite{spielman2004smoothedanalysis} introduced smoothed analysis. This framework adds a small amount of random noise to the input data to model slight imprecisions occurring in the real world. If the expected runtime of an algorithm on these perturbed inputs is efficient, this may explain why the algorithm has a good practical performance. Smoothed analysis has been applied to many areas in computer science \cite{manthey2011smoothedanalysis, spielman2009smoothedanalysis}, in particular to local search algorithms \cite{englert2014smoothedanalysistsp, arthur2011smoothedanalysiskmeans, mathey2015smoothedanalysis}.

  To perform a smoothed analysis of the FLIP algorithm, we add a small amount of random noise to each edge weight of the graph. In this setting, Etscheid and R{\"o}glin \cite{etscheid2017smoothedlocalmaxcut} have shown that the FLIP algorithm runs in smoothed quasi-polynomial time. This means that the expected runtime is bounded by $\phi n^{\cO(\log n)}$ where $\phi$ is a parameter controlling magnitude of the perturbations applied to the inputs. This bound was later improved to $\phi n^{\cO(\smash{\sqrt{\log n}})}$ by \cite{chen2020smoothedlocalmaxcut}. Moreover, in graphs with logarithmic maximum degree and in complete graphs, the FLIP algorithm runs in smoothed polynomial time \cite{elsasser2011localmaxcut, giannakopoulos2022smoothed, angel2017smoothedpolynomiallocalmaxcut}. The best known smoothed runtime bound for the FLIP algorithm in complete graphs is $\cO(\phi n^{7.83})$ \cite{bibak2021smoothedflipmaxcut}.

  There are also local search algorithms for the Max-Cut problem which move more than one vertex across the cut in each step \cite{kernighan1970heuristic, gosti1995approximationmaxcut}. One such algorithm is the $k$-FLIP algorithm which moves up to $k$ vertices in each step. $k$-Opt Local Max-Cut is the corresponding problem of finding a cut which cannot be improved even if we are allowed to move up to $k$ vertices across the cut at once. As already noted, the 2-FLIP algorithm has a smoothed quasi-polynomial runtime in complete graphs \cite{chen2023smoothed}. This result assumes that the pivot rule never moves two vertices if moving just one of them provides a higher improvement.
  
  The proofs of these results show that with a small amount of random noise, the expected number of steps of every possible execution of the local search from any initial cut is at most quasi-polynomial. In contrast, we show with \cref{thm:3flipsuppolyruntime} that this is not true for the 3-FLIP algorithm. Even if a large amount of random noise is added to the graph and if we consider a random initial cut, the runtime of the 3-FLIP algorithm can nevertheless be as large as $2^{\Omega(\sqrt{n})}$. This provides the first example where a local search algorithm for the Max-Cut problem has a superpolynomial smoothed runtime. However, we also note that the long local search sequences in our example contain moves with three vertices where moving just one of them would yield a higher improvement. It therefore remains open whether the 3-FLIP algorithm could have an efficient smoothed runtime for pivot rules that do not make such moves, such as the greedy pivot rule, or for restricted graph classes.

  \subsection{Preliminaries}

  For any integer $k \in \N$, we write $[k] := \set{1, \dots, k}$. If $G = (V, E)$ is a graph, we denote by $N(v)$ the neighbourhood of a vertex $v$. A \defn{cut} of $G$ is a function $\sigma\colon V \to \set{-1, 1}$. This cut partitions the vertices of $G$ into $\set{v \in V : \sigma(v) = 1}$ and $\set{v \in V : \sigma(v) = -1}$. We say that an edge $u v \in E$ \defn{crosses the cut} if $\sigma(u) \neq \sigma(v)$. If $G$ is weighted, we denote its \defn{edge weights} by $X \in \R^E$. The \defn{value of a cut} $\sigma$ is
  \[
    v(\sigma) := \frac{1}{2} \sum_{u v \in E} (1 - \sigma(u) \sigma(v)) X_{u v}.
  \]
  Thus, the Max-Cut problem is the problem of finding a cut $\sigma$ which maximizes $v(\sigma)$.

  A \defn{$k$-flip sequence} is a sequence $L = V_1, \dots, V_\ell$ of subsets $V_i \subs V$ such that $\abs{V_i} \le k$ for all $i \in [\ell]$. $L$ can be interpreted as a sequence of $\ell$ moves whose $i$-th step moves all vertices from $V_i$ across the cut. For some initial cut $\sigma$, we define $\sigma^L$ to be the cut obtained from $\sigma$ by performing all moves from $L$. This means that $\sigma^L$ is defined by $\sigma^L(v) := (-1)^{n(v)} \sigma(v)$ where $n(v)$ denotes the number of occurrences of $v$ in $L$. The \defn{change in the cut value} caused by performing all moves from $L$ is $\impr{\sigma}{L} := v(\sigma^L) - v(\sigma)$. For example, for a single vertex $v$, it can be checked that
  \[
    \impr{\sigma}{v} = \sum_{u \in N(v)} \sigma(u) \sigma(v) X_{u v}.
  \]
  Finally, $L$ is \defn{improving} from $\sigma_0$ if $\impr{\sigma_{i-1}}{V_i} > 0$ for all $i \in [\ell]$ where $\sigma_i$ is defined inductively by $\sigma_i := \sigma_{i-1}^{V_i}$. With these definitions, an execution of the $k$-FLIP algorithm corresponds to an improving $k$-flip sequence ending in a local optimum.

  \section{Runtime of the FLIP algorithm}
  \label{sec:flipruntime}
  
  First, we present a small graph where the FLIP algorithm takes exponential time to terminate. Our example is given by the graph $G_n$ and its initial cut depicted in \cref{fig:expex}. Using $\cO(n)$ vertices, this graph generates exponential improving sequences of length $\Omega(3^n)$. The graph is constructed inductively as follows:
  \begin{itemize}
    \item
      The base case is a graph $F_0$ consisting of a single edge $v_{0,1} v_{0,8}$ with weight $7$ whose endpoints are on opposite sides of the cut.
    \item
      We construct a graph $F_n$ from $F_{n-1}$ by adding a path with eight new vertices $v_{n,1}, \dots, v_{n,8}$ to the graph, with consecutive vertices positioned on alternating sides of the cut and $v_{n,1}$ on the same side of the cut as $v_{n-1,1}$. For all $k \in [7]$, the edge $v_{n,k} v_{n,k+1}$ is assigned a weight of $(7 - 2 \floor{k/2}) \cdot 8^n$.
      
      The new vertices are joined to those of $F_{n-1}$ as follows: $v_{n,2}$, $v_{n,4}$, and $v_{n,6}$ are joined to $v_{n-1,1}$. These edges are assigned a weight of $8^n$ except for the edge $v_{n,4} v_{n-1,1}$ which receives the negative weight $-8^n$. The vertices $v_{n,3}$, $v_{n,5}$, and $v_{n,7}$ are connected to $v_{n-1,8}$ with edges of weight of $1$, expect that the weight of the edge $v_{n,5} v_{n-1,8}$ is $-1$.
    \item
      To obtain $G_n$, two new vertices $w_1$ and $w_2$ are added to $F_n$ with $w_1$ on the same side of the cut as $v_{n,1}$ and $w_2$ on the opposite side. Then, $w_1$ is joined to $v_{n,1}$ with an edge of weight $8^{n+1}$ and $w_2$ is joined to $w_1$ with an edge of weight $2 \cdot 8^{n+1}$.
  \end{itemize}
  In total, $G_n$ has $8 n + 4 \in \cO(n)$ vertices and $13 n + 3 \in \cO(n)$ edges whose weights range (in modulus) from $1$ to $2 \cdot 8^{n+1}$. The maximum degree of $G_n$ is four.
  
  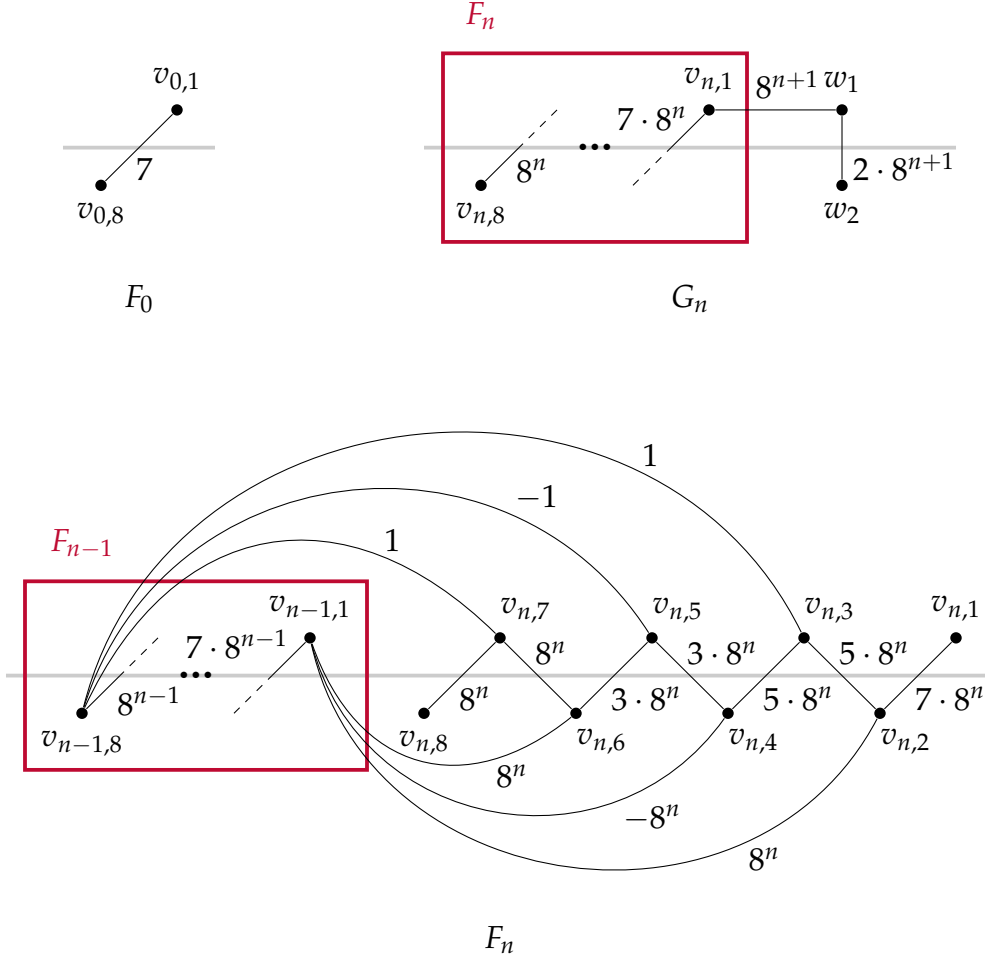
\begin{figure}[t]
    \centering
    \begin{tikzpicture}
      \begin{scope}[shift={(0.25,0)}]
        \draw[black!20,line width=0.05cm] (-0.5,0.5) -- (1.5,0.5);
        \expexbase[name=v1,vertex labels macro=v_{0,#1},edge label=7,state=0]
        \node at (0.5,-1.5) {$F_0$};
      \end{scope}
      \begin{scope}[shift={(5.25,0)}]
        \draw[black!20,line width=0.05cm] (-0.75,0.5) -- (6.25,0.5);
        \draw[structurecolor,line width=0.05cm] (-0.5,-0.75) rectangle (3.5,1.75);
        \node[structurecolor] at (0,2.25) {$F_n$};
        \expexblackbox[name=v1,vertex labels macro=v_{n,#1},edge labels={8^n,7\cdot8^n},state=0]
        \vertex["$w_1$"]{w1}{4.75,1}
        \vertex["$w_2$"below]{w2}{4.75,0}
        \edges{(v11) --["$8^{n+1}$"pos=0.6] (w1) --["$2 \cdot 8^{n+1}$"pos=0.8] (w2)}
        \node at (2.75,-1.5) {$G_n$};
      \end{scope}
      \begin{scope}[shift={(0,-7)}]
        \draw[black!20,line width=0.05cm] (-1,0.5) -- (12,0.5);
        \draw[structurecolor,line width=0.05cm] (-0.75,-0.75) rectangle (3.75,1.75);
        \node[structurecolor] at (0,2.25) {$F_{n-1}$};
        \expexblackbox[name=v1,vertex labels macro=v_{n-1,#1},edge labels={8^{n-1},7\cdot8^{n-1}},state=0]
        \expexgadget[name=v2,shift={4.5,0},vertex labels macro=v_{n,#1},edge labels={8^n,3\cdot8^n,5\cdot8^n,7\cdot8^n},state=0]
        \expexgadetconnect[connect edge label=8^n,state=0]{v1}{v2}
        \node at (5.5,-3) {$F_n$};
      \end{scope}
    \end{tikzpicture}
    \caption{A graph $G_n$ with $\cO(n)$ vertices where the FLIP algorithm runs in time $\Omega(3^n)$.}
    \label{fig:expex}
  \end{figure}

  The idea behind this construction is that the initial configuration of $F_n$ is a local max-cut, but moving the \defn{start vertex} $v_{n,1}$ across this cut will trigger an exponential improving sequence. This sequence will then move the vertices $v_{n,2}$ to $v_{n,8}$ one-by-one across the cut. Moreover, while doing so, it forces the improving sequence of $F_{n-1}$ to be performed three times, once each after the moves of $v_{n,2}$, $v_{n,4}$, and $v_{n,6}$. This results in the exponential growth of the length of that sequence. Because the additional vertices of $G_n$ ensure that moving $v_{n,1}$ across the cut in the first step is improving, $G_n$ will perform this exponential improving sequence. Moreover, we can prove that this is in fact the unique improving sequence of $G_n$. This gives the following result.

  \begin{theorem}
    For the graph $G_n$ and its initial cut depicted in \cref{fig:expex}, the unique improving sequence that ends in a local optimum has length $\Omega(3^n)$.
  \end{theorem}

  \begin{proof}
    Let $L_0 := v_{0,1}, v_{0,8}$ and
    \[
      L_n := v_{n,1}, v_{n,2}, L_{n-1}, v_{n,3}, v_{n,4}, L_{n-1}, v_{n,5}, v_{n,6}, L_{n-1}, v_{n,7}, v_{n,8}.
    \]
    The length of $L_n$ is $\Omega(3^n)$, and a simple induction shows that $L_n$ moves every vertex of $F_n$ an odd number of times across the cut ($w_1$ and $w_2$ remain fixed). We claim that $L_n$ is improving from the initial cut of $G_n$. For this purpose, we will show by induction on $n$ that every step of $L_n$, except for potentially the first step, is improving when starting from the initial cut of $F_n$. This is clearly satisfied for the base case $L_0$. For any $n > 0$, consider the following cases:
    \begin{itemize}
      \item Let $k \in \set{2, 4, 6, 8}$. Then, since the weight of the edge $v_{n,k} v_{n,k-1}$ is higher than the absolute weights of all other edges incident to $v_{n,k}$ combined, moving $v_{n,k}$ across the cut is improving if and only if $v_{n,k}$ and $v_{n,k-1}$ are on the same side of the cut. This condition is satisfied during the move of $v_{n,k}$ in $L_n$ since $v_{n,k}$ and $v_{n,k-1}$ start on opposite sides of the cut but $v_{n,k-1}$ moves across the cut immediately before $v_{n,k}$. Hence, the move of $v_{n,k}$ in $L_n$ is improving.
      \item Let $k \in \set{3, 7}$. In this case, the edges $v_{n,k} v_{n,k-1}$ and $v_{n,k} v_{n,k+1}$ have the same weight while the third edge $v_{n,k} v_{n-1,8}$ has a weight of $1$. So, there are two situations in which moving $v_{n,k}$ across the cut is improving: Either $v_{n,k-1}$ and $v_{n,k+1}$ are both on the same side of the cut as $v_{n,k}$, or this is satisfied by one of these two vertices and additionally the vertex $v_{n-1,8}$.
      
      Because $v_{n,k}$ and $v_{n-1,8}$ start on opposite sides of the cut and $L_{n-1}$ moves $v_{n-1,8}$ an odd number of times across the cut, $v_{n-1,8}$ will be on the same side of the cut as $v_{n,k}$ during the move of $v_{n,k}$ in $L_n$. This is also true for $v_{n,k-1}$ since $v_{n,k-1}$ starts on the opposite side of the cut from $v_{n,k}$ but moves exactly once across the cut before the move of $v_{n,k}$. Hence, the second condition from above is satisfied, and so the move of $v_{n,k}$ in $L_n$ is also improving.
      \item For $v_{n,5}$, the behaviour is almost exactly the same as for $v_{n,3}$ and $v_{n,7}$, except that $v_{n-1,8}$ must be on the opposite side of the cut from $v_{n,5}$ to make the move of $v_{n,5}$ improving. Since this is satisfied in $L_n$, the move of $v_{n,5}$ is also improving.
      \item Whenever $v_{n-1,1}$ moves in $L_n$ as the first step of $L_{n-1}$, it can be checked that either $v_{n,2}$ and $v_{n,6}$ will both be on the same side of the cut as $v_{n-1,1}$, or this is satisfied by one of these two vertices and additionally $v_{n,4}$ is on the opposite side of the cut. In both of these cases, the contributions of the mentioned vertices outweigh all other edges incident to $v_{n-1,1}$ combined, ensuring that the moves of $v_{n-1,1}$ in $L_n$ are improving.
      \item We also have to reverify that the moves of $v_{n-1,8}$ in $L_n$ as part of $L_{n-1}$ are improving because we have connected new edges to that vertex. However, this is no problem: The weight of the edge $v_{n-1,8} v_{n-1,7}$ is still higher than the absolute weights of all other edges incident to $v_{n-1,8}$ combined. As for $v_{n,8}$, we therefore get that the moves of $v_{n-1,8}$ in $L_n$ are improving.
      \item Finally, all other moves of $L_{n-1}$ are improving in $L_n$ by induction since we connected no new edges to these vertices and so the improvements of these moves in $L_n$ when starting from $F_n$ are the same as the improvements of these moves in $L_{n-1}$ when starting from $F_{n-1}$.
    \end{itemize}
    Hence, every move of $L_n$, except for potentially the first move, is improving when starting from the initial cut of $F_n$. Due to the additional vertices and edges from $G_n$, the first move of $L_n$ in $G_n$ will also be improving. Therefore, $L_n$ is an improving sequence of length $\Omega(3^n)$ in $G_n$.

    Lastly, we will show that $L_n$ is the only improving sequence of $G_n$ ending in a local optimum, which concludes the proof of the theorem. For this, we show that each move of $L_n$ is the unique improving move at that time step. Then, every improving sequence must perform exactly the moves from $L_n$, as required. Again, we divide into cases:
    \begin{itemize}
      \item For $k \in \set{2, 4, 6, 8}$, recall that moving $v_{n,k}$ across the cut is improving if and only if $v_{n,k}$ and $v_{n,k-1}$ are on the same side of the cut. However, these two vertices start on opposite sides of cut, and this remains true until $v_{n,k-1}$ moves. Moreover, because $v_{n,k}$ moves immediately after $v_{n,k-1}$, these vertices will again be on opposite sides of the cut after the move of $v_{n,k}$. Hence, $v_{n,k}$ and $v_{n,k-1}$ are only on the same side of the cut during the move of $v_{n,k}$ in $L_n$, and so $v_{n,k}$ cannot move at any other time step. Similar arguments apply to $v_{n-1,8}$.
      \item For $k \in \set{3, 7}$, recall that moving $v_{n,k}$ across the cut is improving if and only if either $v_{n,k-1}$ and $v_{n,k+1}$ are both on the same side of the cut as $v_{n,k}$, or if this is satisfied by one of these two vertices and additionally the vertex $v_{n-1,8}$. It can be checked that the first case never occurs while the second case is only satisfied when $v_{n,k}$ moves in $L_n$. Hence, $v_{n,k}$ cannot move at any other time step than during its move in $L_n$. Similar arguments apply to $v_{n,5}$.
      \item For $v_{n-1,1}$, note that a move of that vertex can only become improving if one of its adjacent vertices has moved. Since we already know that $v_{n-1,1}$ moves immediately after every move of $v_{n,2}$, $v_{n,4}$, or $v_{n,6}$ in $L_n$, this could only be problematic if a move of $v_{n-1,1}$ becomes improving after one of the moves of $v_{n-1,2}$ in $L_n$. However, this is never the case because $v_{n-1,2}$ will always move to the opposite side of the cut from $v_{n-1,1}$ and will therefore only reinforce the current position of $v_{n-1,1}$ on its side of the cut.
      \item Finally, for all other vertices of $F_{n-1}$, we can show by induction that they cannot move at any other time step than during their moves in $L_n$.
    \end{itemize}
    Hence, a move of a vertex of $F_n$ is only improving when this vertex moves in $L_n$. Moreover, a move of $w_1$ or $w_2$ is never improving in $G_n$. This is again due to similar reasons as for $v_{n,2}$, $v_{n,4}$, $v_{n,6}$, and $v_{n,8}$: Because the edge between $w_1$ and $w_2$ outweighs all other edges incident to these vertices, moving one of these two vertices is only improving if they are both on the same side of the cut, but this is never the case. Hence, as required, the moves of $L_n$ are the unique improving moves at their time steps.
  \end{proof}
  
  In particular, because the FLIP algorithm performs an improving sequence until it reaches a local optimum, it needs an exponential number of steps to terminate for the graph $G_n$, and this holds regardless of the chosen pivot rule. This proves \cref{thm:flipruntimeexp}.

  \section{Smoothed runtime of the 3-FLIP algorithm}
  \label{sec:3flipsmoothedruntime}

  Next, we show that the 3-FLIP algorithm can have a superpolynomial smoothed runtime even if we choose all edge weights of the graph uniformly at random. This relies on the fact that with high probability linear combinations of uniform random variables with coefficients $-1$ and $1$ can approximate any value in a large interval up to some exponentially small error. Since the local improvements during the execution of the 3-FLIP algorithm correspond to such linear combinations, this allows us to control these local improvements in graphs whose edge weights are chosen uniformly at random.

  Define the \defn{approximation error} of an interval $I$ by linear combinations of $X_1, \dots, X_n$ as
  \[
    \aperr(I, X_1, \dots, X_n) = \max_{x \in I} \l(\min_{a_1, \dots, a_n \in \set{-1, 1}} \abs*{x - \sum_{i = 1}^n a_i X_i}\r).
  \]

  \begin{lemma}
    \label{lem:denlincombunirandvar}
    For all $b \in \R$ there exists $c > 0$ such that the following holds. If $n \in \N$ is sufficiently large and $X_1, \dots, X_n$ are independent random variables chosen uniformly at random in the interval $[b,b+1]$, then it holds that
    \[
      \pr\l(\aperr\l(\l[-\frac{n}{5},\frac{n}{5}\r], X_1, \dots, X_n\r) < 2^{-cn}\r) \ge 1 - 2^{-cn}.
    \]
  \end{lemma}

  \begin{proof}
    A result of Lueker \cite[Corollary 3.3]{lueker1998exponentiallysmallbounds} implies that there exists $c > 0$ such that $\ex(\aperr([-n/5,n/5], X_1, \dots, X_n)) \le 2^{-2cn}$ if $n$ is sufficiently large. It follows that
    \[
      \pr\l(\aperr\l(\l[-\frac{n}{5},\frac{n}{5}\r], X_1, \dots, X_n\r) \ge 2^{-cn}\r) \le \frac{\ex(\aperr([-n/5,n/5], X_1, \dots, X_n))}{2^{-cn}} \le 2^{-cn}. \qedhere
    \]
  \end{proof}

  We now construct a graph $H_k$ with $n \in \cO(k^2)$ vertices where the 3-FLIP algorithm can have a smoothed runtime of $\Omega(2^k) = 2^{\Omega(\sqrt{n})}$. For any $k \in \N$, let $n_k \in \cO(k)$ be such that $2^{-c n_k / 64} < 3^{-k}$ for the $c$ given by \cref{lem:denlincombunirandvar}. The graph $H_k$, depicted in \cref{fig:tlmaxcutex}, is a disjoint union of the following graphs:
  \begin{itemize}
    \item For every $i \in [k]$, the graph contains two trees $S^{v_i}$ and $S^{w_i}$. These two trees are obtained by taking a copy of $K_{1,2n_k}$ and attaching two new vertices to each leaf. We denote the centers of $S^{v_i}$ and $S^{w_i}$ by $v_i$ and $w_i$ respectively.
  \end{itemize}

  \begin{figure}[t]
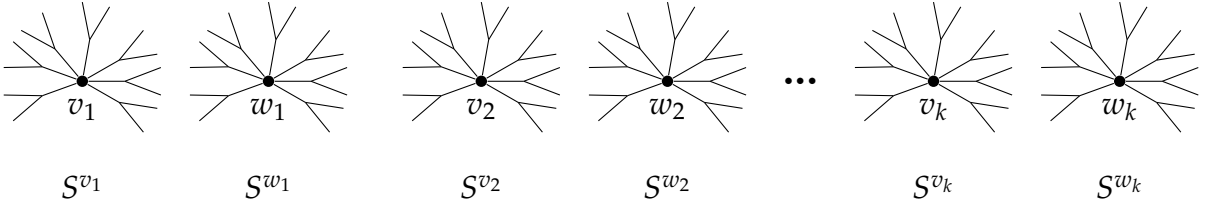

    \centering
    \begin{graph}[scale=0.7]
      \tlmaxcuttree[shift={0,0},state=1,vertex label=v_1]
      \node at (0,-2) {$S^{v_1}$};
      \tlmaxcuttree[shift={3.5,0},state=1,vertex label=w_1]
      \node at (3.5,-2) {$S^{w_1}$};
      \tlmaxcuttree[shift={7.5,0},state=1,vertex label=v_2]
      \node at (7.5,-2) {$S^{v_2}$};
      \tlmaxcuttree[shift={11,0},state=1,vertex label=w_2]
      \node at (11,-2) {$S^{w_2}$};
      \node at (13.5,0) {\tiny$\bullet\bullet\bullet$};
      \tlmaxcuttree[shift={16,0},state=1,vertex label=v_k]
      \node at (16,-2) {$S^{v_k}$};
      \tlmaxcuttree[shift={19.5,0},state=1,vertex label=w_k]
      \node at (19.5,-2) {$S^{w_k}$};
    \end{graph}
    \caption{A graph $H_k$ with $n \in \cO(k^2)$ vertices where the 3-FLIP algorithm can have a smoothed runtime of up to $2^{\Omega(\sqrt{n})}$.}
    \label{fig:tlmaxcutex}
  \end{figure}
  
  Because $S^{v_i}$ and $S^{w_i}$ have $6 n_k + 1 \in \cO(k)$ vertices each, the graph $H_k$ has a total of $2 k (6 n_k + 1) \in \cO(k^2)$ vertices as claimed. We also remark that we could make $H_k$ connected by arbitrarily adding edges between the leaves of the trees of $H_k$. Such edges are irrelevant to the improving 3-flip sequence that we construct.
  
  To study the smoothed complexity of the 3-FLIP algorithm in $H_k$, suppose that the edge weights of $H_k$ are chosen uniformly at random and consider a uniformly random initial cut $\sigma$ of $H_k$. To show that the 3-FLIP algorithm can have a superpolynomial runtime in this setting, we start by moving a subset of the neighbours of $v_i$ and $w_i$ across the cut. This allows us to get a very precise control over the local improvements of a move of $v_i$ or $w_i$ later in the execution of the algorithm. This is given by the following result.

  \begin{lemma}
    Let $b \in \R$ be a real number. Suppose that the edge weights of the graph $H_k$ are chosen uniformly at random in the interval $[b,b+1]$ and that $\sigma$ is a uniformly random cut of $H_k$. Then, with high probability there exists a 3-flip sequence $L$ that is improving from $\sigma$ such that $\abs{\impr{\sigma^L}{u} - 3^{-(i-1)}} < 3^{-k}$ for all $i \in [k]$ and $u \in \set{v_i, w_i}$.
  \end{lemma}

  \begin{proof}
    We may assume that $b \ge -1/2$ since the case $b \le -1/2$ is analogous. Let $i \in [k]$ and $u \in \set{v_i, w_i}$. For each $s \in N(u)$, denote by $s_1, s_2 \in N(s) \sm \set{u}$ the two neighbours of $s$ apart from $u$. Let $E_s$ be the event that $\sigma(s) = \sigma(s_1) = \sigma(s_2)$ and $X_{s s_1} + X_{s s_2} > b + 1$, and let $S_u \subs N(u)$ be the subset of those neighbours $s$ of $u$ where the event $E_s$ occurs. These events are independent with $\pr(E_s) \ge 1/32$. So, a Chernoff bound implies that
    \[
      \pr\l(\abs{S_u} \ge \frac{n_k}{32}\r) = \pr\l(\abs{S_u} \ge \frac{\abs{N(u)}}{64}\r) \ge 1 - 2^{-\Omega(n_k)}.
    \]
    Define
    \[
      D_u := \sum_{s \in N(u) \sm S_u} \sigma(u) \sigma(s).
    \]
    Note that $\sigma(u) \sigma(s)$ is chosen uniformly at random in $\set{-1, 1}$, independent of the event $E_s$. Since $D_u$ is a sum of at most $2 n_k$ such independent random variables, a Chernoff bound implies that
    \[
      \pr\l(\abs{D_u} \le \frac{n_k}{64}\r) \ge 1 - 2^{-\Omega(n_k)}.
    \]
    If $E_u$ is the event that $\abs{S_u} \ge n_k/32$ and $\abs{D_u} \le n_k/64$ are both satisfied, this implies that $\pr(E_u) = 1 - 2^{-\Omega(n_k)}$.
    
    Let $m := \floor{n_k / 128}$ and let $T_u \subs S_u$ be a subset of size $2 m$. If $E_u$ occurs, then we have that $\abs{S_u \sm T_u} \ge n_k/64 \ge \abs{D_u}$. So, the fact that $\abs{N(u) \sm S_u} + \abs{S_u \sm T_u} = 2 (n_k - m)$ is even implies that there exist coefficients $a_{u s} \in \set{-1, 1}$ for $s \in S_u \sm T_u$ such that $D_u + \sum_{s \in S_u \sm T_u} a_{u s} = 0$. Define
    \[
      \Delta_u := \sum_{s \in N(u) \sm S_u} \sigma(u) \sigma(s) X_{u s} + \sum_{s \in S_u \sm T_u} a_{u s} X_{u s}.
    \]
    Note that $X_{u s}$ is distributed uniformly at random in $[b,b+1]$, independent of the events $E_u$ and $E_s$. By pairing up variables with positive and negative coefficients, $\Delta_u$ can be written as a sum of $n_k - m$ independent random variables taking values in $[-1, 1]$, each with expectation $0$, and so $\ex(\Delta_u) = 0$. Therefore, Hoeffding's inequality \cite[Theorem 2]{hoeffding1963probabilityinequalities} implies that
    \[
      \pr\l(\abs{\Delta_u} \le \frac{n_k}{512}\r) \ge 1 - 2^{-\Omega(n_k)}.
    \]
    Let $F_u$ be the event that there exist coefficients $a_{u s} \in \set{-1, 1}$ for $s \in T_u$ such that $\abs{3^{-(i-1)} - \Delta_u - \sum_{s \in T_u} a_{u s} X_{u s}} < 3^{-k}$. If $E_u$ and $\abs{\Delta_u} \le n_k/512$ are both satisfied, then $\abs{3^{-(i-1)} - \Delta_u} \le 1 + n_k/512 \le 2m/5$. Moreover, the random variables $X_{u s}$ for all $s \in T_u$ are independent of these two conditions. So, \cref{lem:denlincombunirandvar} shows that the probability of $F_u$ conditioned on these two events will be at least $1 - 2^{-\Omega(n_k)}$, and thus
    \[
      \pr(F_u) \ge \pr\l(F_u \st E_u \text{ and } \abs{\Delta_u} \le \frac{n_k}{512}\r) \pr\l(E_u \text{ and } \abs{\Delta_u} \le \frac{n_k}{512}\r) \ge 1 - 2^{-\Omega(n_k)}.
    \]
    Applying the union bound shows that with high probability all the events $F_u$ will occur. Consider the case where this happens.

    Let $L$ be the $3$-flip sequence moving every vertex $s \in S_u$ with $\sigma(u) \sigma(s) \neq a_{u s}$ individually across the cut. Because the event $E_s$ occurs for every $s \in S_u$, these moves are improving. Let $\tau := \sigma^L$ be the resulting cut, so $\tau(s) = -\sigma(s)$ for exactly those $s \in S_u$ with $\sigma(u) \sigma(s) \neq a_{u s}$. In particular, this implies that $\tau(u) \tau(s) = a_{u s}$ for all $s \in S_u$ while $\tau(u) \tau(s) = \sigma(u) \sigma(s)$ for all $s \in N(u) \sm S_u$. Therefore,
    \[
      \impr{\tau}{u} = \sum_{s \in N(u)} \tau(u) \tau(s) X_{u s} = \Delta_u + \sum_{s \in T_u} a_{u s} X_{u s},
    \]
    and so we conclude that $\abs{\impr{\tau}{u} - 3^{-(i-1)}} < 3^{-k}$ because the event $F_u$ occurs.
  \end{proof}

  Given these controlled local improvements, we now construct an improving 3-flip sequence of length $\Omega(2^k)$ in $H_k$. This 3-flip sequence acts like a binary counter on the vertices $v_1, \dots, v_k$. Whenever $v_i$ is on the side $\sigma^L(v_i)$ of the cut, this corresponds to a $0$ at position $i$ of the binary counter, while otherwise it corresponds to a $1$. For any $i$, the sequence will first perform a binary counter sequence on $v_{i+1}, \dots, v_k$, will then move $v_i$ across the cut while resetting the positions of $v_{i+1}, \dots, v_k$, and will afterwards again perform a binary counter sequence on $v_{i+1}, \dots, v_k$. Since this construction requires a move of $v_i$ to reset the positions of $v_{i+1}, \dots, v_k$ and these vertices cannot all move in the same step, the vertices $w_{i+1}, \dots, w_k$ are used to pass on the task of resetting the positions of $v_{i+1}, \dots, v_k$. This yields the following result.
  
  \begin{lemma}
    \label{lem:tlmaxcutex}
    Let $\sigma$ be a cut of the graph $H_k$ with $\abs{\impr{\sigma}{u} - 3^{-(i-1)}} < 3^{-k}$ for all $i \in [k]$ and $u \in \set{v_i, w_i}$. Then, there exists at least one 3-flip sequence that is improving from $\sigma$ and has length $\Omega(2^k)$.
  \end{lemma}

  \begin{proof}
    Let $V_i := \set{v_i, \dots, v_k}$, $U_i := V_i \cup  \set{w_i, \dots, w_k}$, and $U := U_1$. If $\tau$ is a cut of $H_k$ satisfying $\restr{\tau}{V \sm U} = \restr{\sigma}{V \sm U}$, we have for $u \in U$ that
    \[
      \impr{\tau}{u} = \sum_{s \in N(u)} \tau(u) \ub{\tau(s)}_{\mathclap{\hspace{0.4cm}= \sigma(s) \text{ since } s \notin U}} X_{u s} = \tau(u) \sigma(u) \sum_{s \in N(u)} \sigma(u) \sigma(s) X_{u s} = \tau(u) \sigma(u) \impr{\sigma}{u}.
    \]
    Hence, as long as we move only vertices from $U$, moving $u$ from the side $\sigma(u)$ to the side $-\sigma(u)$ of the cut changes the cut value by $\impr{\sigma}{u}$ while moving $u$ from  $-\sigma(u)$ to $\sigma(u)$ yields a change of $-\impr{\sigma}{u}$. In particular, if $i < k$, moving $u \in \set{v_i, w_i}$ from the side $\sigma(u)$ to the side $-\sigma(u)$ of the cut while simultaneously moving $v_{i+1}$ and $w_{i+1}$ from $-\sigma(v_{i+1})$ to $\sigma(v_{i+1})$ and $-\sigma(w_{i+1})$ to $\sigma(w_{i+1})$ respectively increases the cut value by
    \begin{equation}
      \impr{\sigma}{u} - \impr{\sigma}{v_{i+1}} - \impr{\sigma}{w_{i+1}} > (3^{-(i-1)} - 3^{-k}) + 2 \cdot (-3^{-i} - 3^{-k}) \ge 0, \label{eq:threeimpr}
    \end{equation}
    implying that such a move is improving.

    We use this to provide an improving 3-flip sequence of length $\Omega(2^k)$ in $H_k$. Consider any cut $\tau$ with $\restr{\tau}{V \sm U} = \restr{\sigma}{V \sm U}$ and $\restr{\tau}{V_i} = \restr{\sigma}{V_i}$. We will show by induction on $k - i$ that there exists a 3-flip sequence $L$ of length $\ge 2^{k-i}$ which is improving from $\tau$ and uses only vertices from $U_i$. The base case $i = k$ is trivially satisfied by the 3-flip sequence $L := (v_k)$ which moves only the vertex $v_k$ across the cut since this move increases the cut value by $\impr{\sigma}{v_k} > 3^{-(k-1)} - 3^{-k} > 0$ and is therefore improving.

    Now, consider an arbitrary $i \in [k-1]$ and assume that the induction hypothesis holds for $i + 1$. Suppose that $\restr{\tau}{V \sm U} = \restr{\sigma}{V \sm U}$ and $\restr{\tau}{V_i} = \restr{\sigma}{V_i}$. Since $\restr{\tau}{V_{i+1}} = \restr{\sigma}{V_{i+1}}$, we can immediately apply the induction hypothesis and get an improving 3-flip sequence $L_1$ of length $\ge 2^{k-(i+1)}$ starting from $\tau$ and using only vertices from $U_{i+1}$. Let $\tau_1 := \tau^{L_1}$ be the resulting cut after the moves from $L_1$.
    
    Next, let $L_2$ be the 3-flip sequence moving every vertex $u \in U_{i+1}$ with $\tau_1(u) = \sigma(u)$ individually to the side $-\sigma(u)$ of the cut. Again, these moves increase the cut value by $\impr{\sigma}{u} > 0$. Moreover, the resulting cut $\tau_2 := \tau_1^{L_2}$ will satisfy $\restr{\tau_2}{U_{i+1}} = -\restr{\sigma}{U_{i+1}}$ while $\tau_2(v_i) = \sigma(v_i)$ since $v_i$ was not yet moved at all. That is, $v_i$ is on the side $\sigma(v_i)$ of the cut while $v_j$ and $w_j$ for $j > i$ are on the sides $-\sigma(v_j)$ and $-\sigma(w_j)$ respectively. Then, we perform the following sequence of moves:
    \[
      L_3 := (v_i, v_{i+1}, w_{i+1}), \ (w_{i+1}, v_{i+2}, w_{i+2}), \ (w_{i+2}, v_{i+3}, w_{i+3}), \ \dots \ , \ (w_{k-1}, v_k, w_k).
    \]
    Since each step of this sequence moves a vertex $u \in \set{v_j, w_j}$ from the side $\sigma(u)$ to the side $-\sigma(u)$ of the cut while moving $v_{j+1}$ and $w_{j+1}$ from $-\sigma(v_{j+1})$ to $\sigma(v_{j+1})$ and $-\sigma(w_{j+1})$ to $\sigma(w_{j+1})$ respectively, we know from \eqref{eq:threeimpr} that all of these moves are improving. Moreover, because $L_3$ moves every vertex from $V_{i+1}$ exactly once across the cut, the resulting cut $\tau_3 := \tau_2^{L_3}$ satisfies $\restr{\tau_3}{V_{i+1}} = -\restr{\tau_2}{V_{i+1}} = \restr{\sigma}{V_{i+1}}$. Hence, we can again apply the induction hypothesis and perform an improving 3-flip sequence $L_4$ of length $\ge 2^{k-(i+1)}$ starting from $\tau_3$ and using only the vertices from $U_{i+1}$.

    Let $L$ be the concatenation of $L_1$, $L_2$, $L_3$, and $L_4$. Since each of these four sequences is improving, $L$ is improving. Moreover, $L$ uses only vertices from $U_i$ since this holds for $L_1$, $L_2$, $L_3$, and $L_4$. Lastly, because the lengths of $L_1$ and $L_4$ are $\ge 2^{k-(i+1)}$, the length of $L$ is $\ge 2 \cdot 2^{k-(i+1)} = 2^{k-i}$. Therefore, the induction hypothesis holds for $i$ and so by induction for all $i \in [k]$. With $i = 1$ and $\tau = \sigma$, this yields a 3-flip sequence of length $\ge 2^{k-1} \in \Omega(2^k)$ which is improving from $\sigma$.
  \end{proof}

  Combining these two lemmas and rescaling the edge weights proves \cref{thm:3flipsuppolyruntime}.
  
  We note that the 3-FLIP algorithm can be easily adapted in this example to avoid the improving 3-flip sequence of superpolynomial length. For instance, this can be achieved by using the greedy pivot rule which always chooses the move maximizing the local improvement in each step. We also note that our example is quite sparse. So, for specific pivot rules and dense graphs, the smoothed complexity of the 3-FLIP algorithm remains open.

  It is also interesting to consider the 2-FLIP algorithm. However, the example from above cannot be adapted to show that the 2-FLIP algorithm would have a superpolynomial smoothed runtime since it would only yield improving 2-flip sequences of length $\Omega(k^2)$. In fact, we can prove the following.
  
  \begin{theorem}
    \label{thm:twolmaxcutoptindex}
    Let $G$ be a graph and $I \subseteq V$ be an independent set. Then, every improving 2-flip sequence of $G$ which uses only vertices from $I$ has a length of $\cO(n^2)$.
  \end{theorem}

  \begin{proof}
    Let $I = \{v_1, \dots, v_k\}$. Since $I$ is an independent set, moving a vertex $v_i$ across the cut will always result in the same change to the cut value. Denote this change by $\Delta_i > 0$ and let $s_i$ be such that moving $v_i$ from the side $s_i$ to $-s_i$ of the cut changes the cut value by $\Delta_i$ while moving $v_i$ from $-s_i$ to $s_i$ changes the cut value by $-\Delta_i$. By sorting the vertices $v_1, \dots, v_k$, we may assume that $\Delta_1 \ge \Delta_2 \ge \dots \ge \Delta_k$.

    Let $L$ be any improving 2-flip sequence of $G$. For any cut $\sigma$ of $G$, assign a token to every vertex $v_i$ with $\sigma(v_i) = s_i$ and consider how these tokens change during the execution of $L$. For this purpose, assume that a single step of $L$ moves two vertices $v_i$ and $v_j$ across the cut where $i < j$. Since this move is improving and $\Delta_i \ge \Delta_j$, this step must move $v_i$ from the side $s_i$ to $-s_i$ of the cut. In particular, $v_i$ loses a token during this step. If $v_j$ gains a token during this move, we imagine the token to be moved from $v_i$ to $v_j$.

    Hence, during every step of $L$, either at least one token gets removed from $G$ or a token moves from some vertex $v_i$ to a vertex $v_j$ with $j > i$. Since there are at most $k$ tokens at the beginning and each of these tokens can move at most $k - 1$ times, it follows that the length of $L$ is at most $k^2 \in \cO(n^2)$.
  \end{proof}

  As the 3-flip sequence from \cref{lem:tlmaxcutex} moved only vertices from an independent set, no small change to this example can yield superquadratic improving 2-flip sequences. We remark that Kaul and West \cite{kaul2008longlocalsearchmaxcut} constructed unweighted graphs with quadratic improving 1-flip sequences. These sequences would still be improving even after adding a small amount of random noise.

  \section{Discussion and open problems}
  \label{sec:openproblems}

  For NP-hard problems, smoothed analysis is an extremely powerful framework for proving that algorithms typically run in polynomial or near-polynomial time. However, the results in this paper show that the effectiveness of this framework can depend on the details of the analysed algorithms. It would be of great interest to find more general conditions under which polynomial runtime can be achieved in this framework.

  Many interesting questions remain open for further investigation.

  \begin{itemize}
    \item Is the smoothed runtime of the FLIP algorithm on arbitrary graphs polynomial? Only quasi-polynomial smoothed runtime bounds are known \cite{etscheid2017smoothedlocalmaxcut, chen2020smoothedlocalmaxcut}.
    \item Is the smoothed runtime of the 2-FLIP algorithm on complete graphs polynomial? Again, only quasi-polynomial runtime bounds are known \cite{chen2023smoothed}.
    \item What is the behaviour of the 2-FLIP algorithm on arbitrary graphs? Boodaghians, Kulkarni, and Mehta \cite{boodaghians2018smoothed} showed that an efficient smoothed runtime of the 2-FLIP algorithm would also prove that the sequential-better-response algorithm for $k$-strategy network coordination games has an efficient smoothed runtime.
    \item We have shown that the smoothed runtime of the 3-FLIP algorithm can be as large as $2^{\Omega(\smash{\sqrt{n}})}$. Is it possible to avoid this by choosing the right pivot rule? It would be very interesting to determine whether the 3-FLIP algorithm has an efficient smoothed runtime for specific pivot rules (for example the greedy pivot rule, or the uniform random pivot rule) or in certain graph classes (like complete graphs). It seems difficult to adapt the example from this paper to these settings. The same questions arise for the $k$-FLIP algorithm with $k > 3$.
  \end{itemize}

  \textbf{Acknowledgements.} We would like to thank an anonymous referee for making us aware of the reference \cite{lueker1998exponentiallysmallbounds}.
  
  {
    \fontsize{11pt}{12pt}
    \selectfont

    \hypersetup{linkcolor=structurecolor}
\newcommand{\etalchar}[1]{$^{#1}$}

  }
\end{document}